\documentclass[runningheads]{llncs}
\usepackage{graphicx}
\usepackage{hyperref}       
\usepackage{url}            

\usepackage{amsmath, bm}
\usepackage{amsfonts}       
\usepackage{subfigure}

\usepackage{multirow}

\makeatletter
\def\cleartheorem#1{%
    \expandafter\let\csname#1\endcsname\relax
    \expandafter\let\csname c@#1\endcsname\relax
}
\def\clearthms#1{ \@for\tname:=#1\do{\cleartheorem\tname} }
\makeatother

\clearthms{claim}

\newtheorem{claim}{Claim}

\newtheorem{new-claim}{Claim}

\newcommand{\fs}{{\bf s}}

\newcommand{\fS}{{\bf S}}
\newcommand{\fu}{{\bf u}}

\newcommand{\bR}{\mathbb{R}}

\newcommand{\cG}{\mathcal{G}}

\newcommand{\cT}{\mathcal{T}}

\begin{document}
\title{Maximal Information Propagation via Lotteries\thanks{
The authors thank Costantinos Daskalakis, Yossi Gilad, Victor Luchangco, and Silvio Micali
for helpful discussions in the early stage of this work, and several anonymous reviewers for helpful comments.
Part of this work was done when Bo Li was an intern at Algorand.
Bo Li is partially supported by the Hong Kong Polytechnic University (No. P0034420).}}
%
%
\author{Jing Chen\inst{1,2} \and Bo Li\inst{3}}
\authorrunning{J. Chen and B. Li.}
%
\institute{Algorand Inc, Boston, MA 02116, USA  \and
Department of Computer Science, Stony Brook University, NY 11794, USA
\\
\email{jing@algorand.com} \and
Department of Computing, Hong Kong Polytechnic University, Hong Kong, China\\
\email{comp-bo.li@polyu.edu.hk}}
\maketitle              
\begin{abstract}
Propagating information to more people through their friends is becoming an increasingly important technology used in domains such as blockchain, advertising, and social media. To incentivize people to broadcast the information, the designer may use a monetary rewarding scheme, which specifies who gets how much, to compensate for the propagation. Several properties are desirable for the rewarding scheme, such as budget feasible, individually rational, incentive compatible and Sybil-proof. In this work, we design a free market with lotteries, where every participant can decide by herself how much of the reward she wants to withhold before propagating to others. We show that in the free market, the participants have a strong incentive to maximally propagate the information and all the above properties are satisfied automatically.

\keywords{Information propagation  \and Nash equilibrium \and Free market design.}
\end{abstract}
\section{Introduction}

Propagating information to more people through their friends is becoming an increasingly important technique used in many fields including advertising, social media \cite{drucker2012simpler} and blockchain \cite{babaioff2012bitcoin,DBLP:conf/marble/ErsoyEL19}.
To incentivize people to broadcast the information, the information holder, i.e., the mechanism designer, may use a monetary rewarding scheme $\mathbf{r} = (r_1,\cdots,r_n)$ to compensate people, where $r_i$ is the reward assigned to player $i$.
The rewarding scheme is expected to satisfy several properties, such as being {\em incentive compatible}, {\em (strongly) budget feasible}, {\em individual rational}, and {\em Sybil-proof}.
Informally, incentive compatibility requires that each player does not decrease her utility by propagating the information,
budget feasibility requires $\sum_i r_i \le B$ when the mechanism designer has a budget of $B$,
strong budget feasibility requires $\sum_i r_i = B$,
individual rationality requires $r_i \ge 0$ for all~$i$,
and Sybil-proofness requires the players do not benefit by making fake copies in the information propagation process.
Accordingly, designing a rewarding scheme to satisfy all or some  of the above properties establishes a large research agenda.
For example, maximal information propagation with budgets is studied in \cite{DBLP:conf/ecai/ShiZSWZ20},
where a rewarding scheme that is incentive compatible, strongly budget feasible and individual rational is proposed.
But the scheme is not Sybil-proof.
A Sybil-proof rewarding scheme is designed in \cite{chen2010mobicent}, but it is not strongly budget feasible,
where a small portion of the claimed reward is distributed among the players.
In this work, instead of designing a centralized rewarding scheme,
we propose a free market with lotteries where everyone can decide how much of the reward she wants to withhold before propagating to others.
We show that in such a market, the players have strong incentives to fully propagate the information and all above properties are satisfied automatically.

To illustrate our design, let us first consider a toy game.
Initially, a seller sends her promotion information to a small number of players she is able to reach, denoted by $N$ and $n = |N|$ who are called {\em aware players}.
The information is associated with a lottery such that one {\em winner} among aware players will be uniformly and randomly selected to get a reward normalized to \$1. 
When nobody propagates the information, everyone's expected reward is $1/n$.
For player $i \in N$, she has a set of friends $F_i$ such that $f_i = |F_i| \ge 1$ and $N\cap F_i = \emptyset$. 
If $i$ signs an agreement with $F_i$ such that if anyone in $F_i$ is selected to be the winner, the reward is given $i$, 
then by propagating the information to $F_i$, player $i$'s reward, which equals to the probability that the winner is selected from $F_i\cup \{i\}$, is 
\begin{align} \label{eq:intro}
    \frac{1+f_i}{n+f_i} > \frac{1}{n}. 
\end{align}
That is $i$ can increase her reward by propagating the information to her friends,
and thus no propagation for $N$ is not a Nash equilibrium.
Actually, since Inequality \ref{eq:intro} holds for any $n \ge 2$, everyone's dominate strategy is to propagate the formation to their friends.
But to what extent will they propagate?

Assuming all players in $N$ except $i$ inform the information to their friends and withhold the complete reward, let us see what $i$ will do.
If $i$ does not withhold the complete reward but shares a small amount, say $0<c<1$, with $F_i$,
$i$ can again improve her utility.
For $j \in F_i$, let $F_j$ be $j$'s friends who are not in the game. 
Supposing all $j\in F_j$ adopts the same strategy with $N\setminus\{i\}$, i.e., fully propagating the information to their friends and withholding the complete reward $c$, we compare two cases for player $i$: (1) $i$ does not leave any reward to $j$, and (2) $i$ leaves $c$ to $j$.
To ease the notation, suppose there are $n'$ players who are in the game except $j$'s friends $F_j$ for every $j \in F_i$.
It is not hard to see that for case (1), player $i$'s reward is $(1+f_i)/n'$ and for case (2), her reward is
\[
\frac{1 + (1-c) \cdot \sum_{j \in F_i} f_j}{\sum_{j \in F_i} f_j+ n'} > \frac{1+f_i}{n'},
\]
as long as for all $j \in F_i$,
\[
f_j > \frac{n+1}{n(1-c) - 1} \to 1 \text{ if $c \to 0$.}
\]
Thus we conclude that withholding the complete reward from propagation is not a Nash equilibrium,
and every player wants to leave partial reward to her friends for incentivizing them to further propagate the information.

We note that the initial $n$ aware players are in the {\em Prisoner's dilemma}.
If they do not propagate the information, each of them has expected utility of $1/n$.
However, it is dominant for each of them to refer their friends by sacrificing partial reward, resulting the expected utility strictly smaller than $1/n$.
This dilemma actually motivates the information propagation in the free market.

A more practical example for the above scenario is the mining game of Bitcoin \cite{bitcoin2008}.
In Bitcoin, when a user makes a transaction (the information sender), she wants the transaction to be broadcasted
(with other necessary information such as account information, transfer amount, crypto signature and etc)
in the network so that the miners can authorize the validity of the transaction and assemble newly
verified transactions into blocks.
The miners compete to propose their blocks to the public chain by solving a computationally hard puzzle,
and the winning probability is proportional to the share of each miner's computation power in the system.
Accordingly, the transaction maker can reward the winning miner who authorizes her transaction a fixed amount of Bitcoins.
At first glance, it seems that the miners may not want to broadcast the transaction
since only the miners who know the transaction can be rewarded.
A centralized rewarding scheme is proposed in \cite{babaioff2012bitcoin} which
not only rewards the winner but also other miners who helped broadcast the transaction.
By carefully designing who gets much, their scheme is Sybil-proof in tree networks.
However, as we will show in this work, the design of free market with lotteries
automatically incentivizes the miners to propagate the transaction and satisfies all other desired properties as well.
Thus the take-home message of this work is that
\begin{quote}
{\em  the mechanism designer does not need to specify each player's reward,
 the market itself already provides incentives for maximal propagation.}
\end{quote}

\subsection{Our Contribution}

We model the problem, and the Bitcoin example, as an information propagation game in a free market,
where a sender has a single piece of information to be broadcasted.
For simplicity, we first assume the players are connected by a complete $d$-ary tree with $d \ge 3$, and all players' winning probabilities are the same.
If there is an edge between two players, they are friends and one can be informed the information by the other.
A strategy profile is called {\em full propagation} if every player withholds a minimum charge
and leaves the remaining reward to all her friends, so that a maximum number of people could be aware of the information.
We show that full propagation forms a Nash equilibrium which satisfies extra properties and thus is more stable than an arbitrary one.

First, full propagation is robust to collective deviations of friends, i.e., {\em connected coalition-proof} \cite{bernheim1987coalition}.
In the seminal work by Myerson \cite{DBLP:journals/mor/Myerson77}, the communication game is proposed
where the network on players represents the possible communication between them.
Originally, it is associated with cooperative games and only coalitions formed by players who can communicate with each other
(i.e., connected subgraphs in the network) are concerned.
We adapt this principle to our problem and show that any deviation of a connected coalition of players from full propagation makes at least one of them worse off.

Second, full propagation survives in any order of iterative elimination of dominated strategies,
and uniquely survives in a particular one.
This result coincides with and generalizes the result of \cite{babaioff2012bitcoin},
where a centralized rewarding scheme is designed.
In a centralized rewarding scheme, a player can only misbehave by withholding the information and claiming Sybil copies,
while in a free market a player can arbitrarily claim how much of the reward she wants to deduct before propagating to others.
We formally discuss the difference between our work and \cite{babaioff2012bitcoin} at the end of Section \ref{sec:tree}.
Recall that strategy $s$ (weakly) dominating $s'$, denoted by $s \preceq s'$,
means choosing $s$ always gives at least as good an outcome as choosing $s'$,
no matter what the other players do,
and there is at least one profile of opponents' actions for which $s$ gives a strictly better outcome than $s'$.
We prove that full propagation is the unique strategy profile that survives in an interval-based monotone elimination
of dominated strategies, coinciding with the players' reasoning process as illustrated in the introduction.

Our main results can be summarized as follows.

\medskip

\noindent\textbf{Main Result 1.} (Theorems \ref{thm:coordinate} and \ref{thm:coordinate:elimination}) {\em
In the tree-structured free market with lotteries, full propagation achieves maximal propagation and satisfies the following:
\begin{enumerate}
\item Full propagation is a Nash equilibrium;
\item Any deviation of a coalition of friends hurts at least one of them;
\item Full propagation survives in any order of iterative elimination of dominated strategies;
\item There is an order of iterative elimination of dominated strategies such that full propagation is the  unique surviving strategy.
\end{enumerate}
}

\medskip

We then extend the above results to non-tree networks.
For arbitrary networks, we introduce stronger relationships than friends, {\em good friends} and {\em best friends}, using shortest paths from the information sender to players.
Although properties 3 and 4 in Main Result 1 do not hold,
we show that if every player has at least three good friends, full propagation is a Nash equilibrium that is also
connected coalition-proof on the induced {\em good-friendship subgraph}.
It is noted that $d$-ary tree with $d \ge 3$ is a special case satisfying this condition.

\medskip

\noindent\textbf{Main Result 2.}  (Theorem \ref{thm:nontree}) {\em
If every player has at least 3 good friends, then full propagation is a Nash equilibrium that is connected coalition-proof on the good-friendship subgraph.
}
\medskip

In conclusion, if the network is well structured, the free market with lotteries is incentive compatible where players are willing to fully propagate the information.
It is not hard to check that our scheme also satisfies other properties mentioned in the introduction.
It is strongly budget feasible and individual rational, as the full reward will be given to the players and none of them needs to pay.
Our scheme is also Sybil-proof, as in our game, the players are required to provide certificates on how many people they have referred to.
For example, in Bitcoin each miner needs to contribute their computing power on authorizing the transaction;
and in sales promotion, only people who physically appear in the market can be selected as the winner.
Therefore, as everyone in our game can arbitrarily decide how much reward she wants to get back from friends,
making Sybil identities is the same as withholding a higher fraction of the reward.

Finally, it is an interesting future direction to generalize our results to random networks.
We believe full propagation brings players high utility in a broader class of networks.
To shed more light in this direction, we conduct experiments in
Section \ref{sec:experiments} for our scheme under general random networks.
In all experiments, full propagation brings a player the maximum utility.
Moreover, the utility gap between full propagation and other strategies is actually large.

\subsection{Related Works}

Our work is partially motivated by the extensive study of incentivizing relays in a blockchain network to propagate transactions.
While this has been studied in the literature,
most of them focus on centralized algorithms where each relay's reward is fixed and decided by the algorithm.
With these algorithms, to gain higher utility, a strategic relay may claim fake copies, i.e., {\em Sybil attack} \cite{douceur2002sybil}.
Accordingly, in works such as \cite{babaioff2012bitcoin,ersoy2018transaction}, {\em Sybil-proof} algorithms are studied.
Since in our free market each relay is able to arbitrarily decide how much reward she is willing to withhold, 
it is superfluous for them to make fake copies, and thus, our results directly imply Sybil-proofness.
The free market ideas have also been discussed independently in \cite{abraham2016solidus} and \cite{chenRelay}  without a systematic analysis.

Our work also aligns with the fundamental study of what the optimal way is to reward miners for their work on authorizing transactions.
Currently, the most popular rewarding scheme, as adopted by Bitcoin,
is to reward miners proportionally to their share of the total contributed computational power.
As shown in \cite{chen2019axiomatic}, the proportional allocation rule is the unique rule that is simultaneously
non-negative, budget-balanced, symmetric, Sybil-proof, and collusion-proof.
In reality, however, to earn steady rewards, miners pool themselves together,
and the pools are vulnerable to security attacks,
such as selfish mining attack \cite{Eyal2014selfishmining,kiayias2016blockchain,koutsoupias2019blockchain,marmolejo2019competing},
block withholding attack \cite{Rosenfeld2011bwa,schrijvers2016incentive},
and denial of service attack \cite{johnson2014game}.
Cooperative games are used in \cite{lewenberg2015bitcoin} to show that
under high transaction loads, it is hard for managers to distribute rewards in a stable way.

Outside the scope of blockchain, information propagation has also been studied
in multi-level marketing \cite{emek2011mechanisms,drucker2012simpler} and
query incentive networks \cite{kleinberg2005query,arcaute2007threshold,chen2010mobicent,cebrian2012finding,chen2013sybil}.
There are major differences between the transaction propagation in a blockchain network and the
query retrieval in peer-to-peer network.
The players in a query incentive network do not compete with the ones who forwarded the message to them,
and cannot generate an answer that they do not have.
Whereas in a blockchain network, 
every aware player is a potential authorizer with probability proportional to their stakes or computational powers.

\section{Game Theoretic Model and Preliminaries}
\label{sec:game_model}


For technical simplicity, we first assume all players are connected by a tree $\cG = (V,E)$. 
Let $s$ be the root of $\cG $ who is the sender of the information.
Note that $s$ is not a player in our game.
Suppose $s$ has $f$ children and thus excluding $s$ from $G$, there are $f$ subtrees denoted by $\{T_1, \cdots, T_{f}\}$.
When there is no confusion, we also use each $T_i$ to denote the set of nodes in it.
Assume all these subtrees are complete $d$-ary and $f\ge d\ge 3$.
Call $N = V \setminus \{s\}$ the set of {\em players} and denote by $n = |N|$.
To make the players distinguishable from the sender, we stop calling $s$ the root and only call her {\em sender}.
Instead, we call the roots of these subtrees {\em roots on depth 0}, who are the initial players of the game, as shown in the following figure.
In a similar fashion, the children of these subtree roots are viewed to be on depth 1 and so on.
For each node $i \in V$, let $NB_i$ be the set of her children in $\cG$, and thus $NB_s$ contains all initial players.


\begin{figure}[htbp]
\begin{center}
\setlength{\unitlength}{0.7cm}
\begin{picture}(4.5,4.2)(-2.5,-0.25)

\put(-1, 4){\circle*{0.15}}
\multiput(-0.6,4)(0.2,0){31}{\line(1,0){0.15}}
\put(6, 3.9){$\mbox{Sender }s$}
\multiput(3.4, 2.7)(0.2,0){11}{\line(1,0){0.15}}
\put(6, 2.6){$\mbox{Depth }0$}
\put(6.05, 1.9){$\mbox{(Roots)}$}
\multiput(3.4, 1)(0.2,0){11}{\line(1,0){0.15}}
\put(6, 0.9){$\mbox{Depth }1$}
\multiput(6.6,0.2)(0,.2){3}{.}

\put(-7.3, 2){$T_1$}
\put(-1, 4){\line(-4, -1){5}}
\put(-6, 2.75){\circle*{0.15}}
\put(-6, 2.75){\line(-3, -5){1.05}}
\put(-6, 2.75){\line(3, -5){1.05}}
\put(-6, 2.75){\line(0, -1){1.8}}
\put(-7.05, 1){\circle*{0.15}}
\put(-6, 1){\circle*{0.15}}
\put(-4.95, 1){\circle*{0.15}}

\put(-7.05, 1){\line(0, -1){1.1}}
\put(-7.05, 1){\line(-1, -3){0.35}}
\put(-7.05, 1){\line(1, -3){0.35}}

\put(-6, 1){\line(0, -1){1.1}}
\put(-6, 1){\line(-1, -3){0.35}}
\put(-6, 1){\line(1, -3){0.35}}

\put(-4.95, 1){\line(0, -1){1.1}}
\put(-4.95, 1){\line(-1, -3){0.35}}
\put(-4.95, 1){\line(1, -3){0.35}}

\put(-4.2, 2){$T_2$}
\put(-1, 4){\line(-3, -2){1.9}}
\put(-2.9, 2.75){\circle*{0.15}}
\put(-2.9, 2.75){\line(-3, -5){1.05}}
\put(-2.9, 2.75){\line(3, -5){1.05}}
\put(-2.9, 2.75){\line(0, -1){1.8}}
\put(-3.95, 1){\circle*{0.15}}
\put(-2.9, 1){\circle*{0.15}}
\put(-1.85, 1){\circle*{0.15}}

\put(-3.95, 1){\line(0, -1){1.1}}
\put(-3.95, 1){\line(-1, -3){0.35}}
\put(-3.95, 1){\line(1, -3){0.35}}

\put(-2.9, 1){\line(0, -1){1.1}}
\put(-2.9, 1){\line(-1, -3){0.35}}
\put(-2.9, 1){\line(1, -3){0.35}}

\put(-1.85, 1){\line(0, -1){1.1}}
\put(-1.85, 1){\line(-1, -3){0.35}}
\put(-1.85, 1){\line(1, -3){0.35}}

\put(-1.1, 2){$T_3$}
\put(-1, 4){\line(1, -1){1.2}}
\put(0.2, 2.75){\circle*{0.15}}
\put(0.2, 2.75){\line(-3, -5){1.05}}
\put(0.2, 2.75){\line(3, -5){1.05}}
\put(0.2, 2.75){\line(0, -1){1.8}}
\put(-0.85, 1){\circle*{0.15}}
\put(0.2, 1){\circle*{0.15}}
\put(1.25, 1){\circle*{0.15}}

\put(-0.85, 1){\line(0, -1){1.1}}
\put(-0.85, 1){\line(-1, -3){0.35}}
\put(-0.85, 1){\line(1, -3){0.35}}

\put(0.2, 1){\line(0, -1){1.1}}
\put(0.2, 1){\line(-1, -3){0.35}}
\put(0.2, 1){\line(1, -3){0.35}}

\put(1.25, 1){\line(0, -1){1.1}}
\put(1.25, 1){\line(-1, -3){0.35}}
\put(1.25, 1){\line(1, -3){0.35}}

\put(-1, 4){\line(3, -1){3.8}}

\multiput(2.8,1.7)(0,.2){3}{.}

\multiput(2.8,0.2)(0,.2){3}{.}

\end{picture}
\caption{An illustration of the tree structure.}

\vspace{-10pt}
\end{center}
\end{figure}
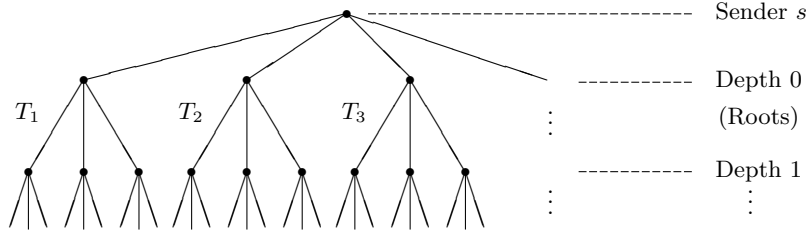

We next define the propagation game $\Gamma  = (N, \fS, \fu)$ on $\cG$.
Without loss of generality, we assume the initial reward set by the sender is 1.
Denote by $S_i$ the strategy space of each player $i\in N$.
Let $\fS = S_1\times\cdots\times S_n$ be the strategy profiles.
Any strategy $s_i \in S_i$ is a mapping from $\bR^+$ to $(\bR^+\cup\{\bot\})^{|NB_i| }$, where $\bR^+ = [0, +\infty)$.
That is, upon receiving the information with remaining reward $x_i$,
player $i$ decides how much reward $s_i(x_i)_j$ to leave to her child $j$ by sending the information to $j$
or not to inform $j$ if $s_i(x_i)_j = \bot$. Denote by $s_i(x_i) = (s_i(x_i)_1, \cdots, s_i(x_i)_{|NB_i|})$.
A {\em feasible strategy} is that for all $j\in NB_i$, if $x_i \ge x_{\min}$,
$0\le s_i(x_i)_j \le x_i - x_{\min}$ or $s_i(x_i)_j = \bot$;
otherwise, $s_i(x_i)_j \equiv \bot$.
Here $x_{\min}>0$ is a sufficiently small number prefixed by the system.
Denote by $\fs = s_1\times\cdots\times s_n$ a strategy profile.
For any strategy profile $\fs\in \fS$ and an initial reward, there is a fixed set of players who can be informed the information, called {\em aware players} and denoted by $N^*(\fs)$.
When $\fs$ is clear in the context, we simply write $N^*$.
To avoid cumbersome calculations, all aware players will get the reward with equal probability of $\frac{1}{|N^*|}$,
and the winner needs to share the reward with her ancestors from her to the sender as committed.
Thus, a player $i$'s (expected) utility consists of two parts: {\em authorizing utility} and {\em referring utility},
\begin{equation*} 
u_i(\fs; x_i) = x_i \cdot \frac{1}{|N^*|}  + \sum_{j\in NB_i, s_i(x_i)_j \neq \bot}(x_i - s_i(x_i)_j) \cdot \frac{n_{ij}}{|N^*|}, 
\end{equation*}
where $n_{ij}$ is the number of aware players following $j$ and including $j$.

Next we formulate the intuition in Inequality \ref{eq:intro} by the following lemma, which shows it is always (strictly) better for a player to propagate the information.

\begin{lemma}\label{lem:no_bot}
For any player $i$,  let $\fs_{-i}$ be $-i$'s strategy profile such that the remaining reward $x_i$ to $i$ is at least  $x_{\min}$.
Let $s_i(x_i) = (z_1, \cdots, z_j, \cdots, z_{|NB_i|})$ be a strategy with $z_j = \bot$ for some $j$
and $s'_i(x_i) = (z_1, \cdots, z'_j, \cdots, z_{|NB_i|})$ be a new strategy by changing $j$'s action from $\bot$ to $0$.
Then $u_i(s_i, \fs_{-i}) < u_i(s'_i, \fs_{-i})$.
\end{lemma}

We prove  Lemma \ref{lem:no_bot} in the appendix.
By Lemma \ref{lem:no_bot}, every aware player does not get hurt by propagating the information to all her neighbours,
and thus in what follows, we assume  without loss of generality $s_i(x_i)_j \ge 0$ for any $x_i \ge x_{\min}$ and $j\in NB_i$.
Moreover, in the following sections, we will see that the aware players actually want to maximally propagate the information.
A feasible strategy $s_i$ is called {\em full propagation} (FP for short) if $s_i(x_i)_j = (k-1) x_{\min}$ for all $j$
and $k x_{\min} \le x_i < (k+1)x_{\min}$ with $k \ge 1$.
That is in an FP strategy, a player wants to inform all her neighbours by leaving the maximal reward to them.
A feasible strategy profile $\fs=(s_i)_{i\in N}$ is called FP if every $s_i$ is FP.
Denote by $\fs^*= (s_i^*)_{i\in N}$ the FP strategy profile.

\section{Main Results}
\label{sec:tree}

To ease formulas, we define the following notations.
For each subtree $T$, let
\[
G(k) =  \sum_{j=1}^k d^{j-1} = \frac{d^k - 1}{d-1}
\]
be the number of nodes from depth 0 to depth $k-1$.
Note that when $d \ge 3$ and $k\ge 1$, $G(k) \ge 2k - 1$.
Moreover, for any $k \ge 1$,
\begin{align}
(f-1)G(k+1) &\ge d^{k+1} - 1 > d(d^{k} - 1) \ge \frac{d+1}{d-1}(d^{k} - 1) \nonumber\\
& =  (d+1)G(k) \ge dG(k) + 2k - 1. \label{eq:G(k)>=2k - 1}
\end{align}
Denote by $H \cdot x_{\min} = 1$ the sender's initial reward, where $H > 1$ is an integer.
Then $H$ is the maximal depth that the information can reach.
For any player $i$ in depth $j \ge 0$,
no matter what $i$'s previous players do,
$i$ is not able to receive the information with remaining reward more than $(H - j) \cdot x_{\min}$
since each of $i$'s ancestors needs to withhold at least $x_{\min}$.
Moreover, all the users in and below depth $H$ are not considered as strategic players.
Finally, under FP strategies, the number of aware players is $f\cdot G(H+1)$.

\subsection{Technical Lemmas}
We next introduce two technical lemmas which are crucial to prove the main results.
For a subtree $T$, let $\pi_0(T)$ be the number of players outside of $T$ who are aware of the information.
If $T$ is clear from the context, we write $\pi_0$ for short.

\begin{lemma}\label{lem:tree:descendant}
Consider the sub-game induced on a single subtree $T$, where $r$ is the root player and $g$ is  one of $r$'s descendants.
Assume $r$ receives the information with reward $x_r$ and $k\cdot x_{\min} \le x_r < (k+1)\cdot x_{\min}$ for some $k\ge 1$.
If $d \ge 3$ and $\pi_0 \ge d\cdot G(k) + 2k - 1$,
$r$'s utility is maximized when $g$ plays FP strategy, taking the actions of the others as given.
\end{lemma}

\begin{proof}
Note that if $x_r < 2x_{\min}$ or the information cannot reach $g$, $r$'s utility does not depend on $g$'s action,
thus the statement is trivially true. 
In the following we assume $x_r \ge 2x_{\min}$. 

Taking the actions of the players (including $r$) except $g$ as given, 
let $\Delta_j$ be the total number of $r$'s referred players from $r$'s child $j$ and $\Delta = \sum_{j\in NB_r} \Delta_j$.
Let $g'$ be $g$'s ancestor who is $r$'s direct child or $g' = g$ when $g$ is $r$'s child.
Thus $g$'s strategy can only change $\Delta_{g'}$, which is the single variable of $r$'s utility.
Formally, $r$'s utility can be written as (assuming $\pi_0$ includes $r$ to simplify notions)
\begin{align*}
u_r(\Delta_{g'}) & = x_r \frac{1}{\pi_0 + \Delta} + \sum_{j \in NB_r} (x_r - s_r(x_r)_j) \frac{\Delta_j}{\pi_0 + \Delta}\\
&  = \frac{x_r + \sum_{j \in NB_r} (x_r - s_r(x_r)_j)  \Delta_j}{\pi_0 + \Delta}.
\end{align*}
Calculating the derivative of $u_r(\Delta_{g'})$, we have
\begin{eqnarray*}
u'_r(\Delta_{g'}) & = &\dfrac{(x_r - s_r(x_r)_{g'})(\pi_0 + \Delta) - (x_r + \sum_{j \in NB_r} (x_r - s_r(x_r)_j)  \Delta_j)}{(\pi_0 + \Delta)^2}  \\
& = & \dfrac{(x_r - s_r(x_r)_{g'})\pi_0 - x_r - \sum_{j \neq g'}  (s_r(x_r)_{g'} - s_r(x_r)_j) \Delta_j}{(\pi_0 + \Delta)^2} \\
& \ge & \dfrac{x_{\min}\pi_0  - x_r - \sum_{j \neq g'}  (x_r- s_r(x_r)_{j})\Delta_j }{(\pi_0 + \Delta)^2},
\end{eqnarray*}
where the inequality is because $s_r(x_r)_j \le x_r - x_{\min}$ for any $j$. 
Let
\begin{eqnarray*}
f_j(y_j) &  =  & (k + 1 - y_{j})G(\lfloor y_j \rfloor) = (k + 1 - y_{j}) \frac{d^{\lfloor y_j \rfloor} -1}{d-1}\\
	   & \le &  (k + 1 - y_{j}) \frac{d^{y_j} -1}{d-1}.
\end{eqnarray*}

\begin{claim} \label{claim:f:bar}
For $d \ge 3$ and $0 \le y_j < k$, $\bar{f}_j (y_j) = (k + 1 - y_{j}) \cdot \frac{d^{y_j} -1}{d-1}$ monotone increases with respect to $y_j$.
\end{claim}
To prove the above claim, it suffices to calculate the derivative of $\bar{f}_j (y_j)$, and we omit the details.
By Claim \ref{claim:f:bar}, $f_j(y_j) \le \bar{f}_j (y_j) \le G(k)$ for any $0 \le y_j < k$.
Thus,
\begin{eqnarray*}
u'_r(\Delta_{g'})
& \ge & \dfrac{x_{\min}\pi_0  - x_r - \sum_{j \neq g'}  (x_r- s_r(x_r)_{j})\Delta_j }{(\pi_0 + \Delta)^2} \\
& \ge & \dfrac{x_{\min}\pi_0  - x_r - G(k) (d-1) x_{\min}}{(\pi_0 + \Delta)^2} \\
& \ge & \dfrac{x_{\min}(\pi_0  - k - 1 - G(k) (d-1) )}{(\pi_0 + \Delta)^2} \ge 0.
\end{eqnarray*}
The last inequality is because $\pi_0 \ge dG(k) + 2k - 1$ and $k \ge 2$.

In conclusion, if $g$ receives the information with remaining reward at least $x_{\min}$, $u'_r(\Delta_{g'}) > 0$,
which means $r$'s utility is maximized when $g$ plays FP strategy, which finishes the proof. \qed
\end{proof}

Lemma \ref{lem:tree:descendant} implies that, if the information is known to a sufficiently large number of players,
for an arbitrary player, her utility is maximized when all her descendants play FP strategies.
Next we show that the other direction of~Lemma~\ref{lem:tree:descendant} also holds: 
for an arbitrary player,
if the information is aware to a sufficiently large number of players and all her descendants play FP strategies,
her utility is maximized when she plays FP strategy.

Fixing a tree $T$, let $i$ be some player in depth $j$ of $T$.
Assume the remaining reward $i$ has received is $kx_{\min} \le x_i < (k+1)x_{\min}$ for some $0 \le k \le H-j-1$, and all her descendants play FP strategies. Given any strategy $s_i(x_i)$, define
$$s'_i(x_i) = \left(\lfloor\frac{s_i(x_i)_1}{x_{\min}}\rfloor \cdot x_{\min}, \cdots, \lfloor\frac{s_i(x_i)_d}{x_{\min}}\rfloor \cdot x_{\min} \right).$$
Note that given her descendants playing FP strategies, $s'_i(x_i)$ brings $i$ utility at least as much as $s_i(x_i)$ does.
The set of all possible  $s'_i(x_i)$ is called {\em reasonable strategies}.
Thus given all $i$'s descendants playing FP strategies, to study $i$'s best response, it suffices to consider reasonable strategies.
For convenience, we use $(d_0, d_1, \cdots, d_k)$ to represent a reasonable strategy,
where $d_l \in [d]$ means $i$ selects $d_l$ children to propagate to next $l$ depths by leaving $(l-1)x_{\min}$ to these children,
and $d_0$ means $i$ selects $d_0$ children to not propagate. 
Thus $\sum_{l = 0}^{k} d_l = d$.
Note that it does not matter which $d_l$ children are selected since all $i$'s children are symmetric.
In the following, we use $u_i((d_0, d_1, \cdots, d_k); x_i)$ to denote $i$'s utility when she receives reward $x_i$
and her action is $(d_0, d_1, \cdots, d_k)$, given all other players adopting FP strategies.

\begin{lemma}\label{lem:tree:induction}
Consider the sub-game on a single subtree $T$, where $r$ is the root player.
Assume $r$ receives the information with reward $x_r$ and $k\cdot x_{\min} \le x_r < (k+1) \cdot x_{\min}$ for some $k\ge 1$.
If $d \ge 3$, $\pi_0 \ge d\cdot G(k) + 2k - 1$, and all $r$'s descendants adopting  FP strategies,
for any reasonable strategy $(d_0, d_1, \cdots, d_k)$, $r$'s utility increases by moving a unit from $0 \le l < k$ to $l+1$.
Formally, if for some $0 \le l <k$ such that $d_l > 0$, then
\[
u_r((d_0, \cdots, d_l, \cdots, d_k); x_r) <  u_r((d_0, \cdots, d_l-1, d_{l+1}+1, \cdots, d_k); x_r).
\]
\end{lemma}

\begin{proof}

To simplify our notions, we ignore the index $r$ for the root player, and let $x$ be the propagation reward that the root receives.
When $kx_{\min} \le x < (k+1)x_{\min}$ and $k \ge 1$, the root is able to leave a proper reward to the players in depth 1
so that the information can be reached to at most depth $k$.
When it is convenient, we write $x = (k + \epsilon) x_{\min}$ and $0 \le \epsilon < 1$.
Given any strategy $(d_0, \cdots, d_k)$ such that $\sum_{j = 0}^{k} d_j = d$ and $d_i > 0$ for some $0 \le i < k$,
we compare $u((d_0, \cdots, d_l, \cdots, d_k); x)$ and $u((d_0, \cdots, d_l-1, d_{l+1}+1, \cdots, d_k); x)$.

If $l = 0$, Lemma \ref{lem:tree:induction} degenerates to Lemma \ref{lem:no_bot} restricted to trees,
which is trivially true.
Thus in the following, we assume $l \ge 1$.

If $ x_{\min} \le x < 2x_{\min}$, the remaining reward a player in depth 1 receives is always smaller than $x_{\min}$;
that is there is no way the information can reach depth 2.
Let $1 < d_0 \le d$ and $d_1 = d - d_0$.
It is not hard to check that by notifying one more player,
$$
u((d_0, d_1); x) = \frac{(d_1 + 1)x}{\pi_0 + 1 + d_1} < \frac{(d_1 + 2)x}{\pi_0  + 1 + d_1 + 1} = u((d_0 - 1, d_1 + 1); x),
$$
where the inequality is because $\pi_0 \ge F(1) > 2$.

Next we assume $k \ge 2$. Denote
$$ Q  = x + \sum_{j = 1}^{k} d_j G(j) (x- (j-1) x_{\min}),$$
and
$$ W = \pi_0 + 1 + \sum_{j = 1}^{k} d_j G(j).$$
Thus the utility for $(d_0, \cdots, d_l, \cdots, d_k)$ is
$$
U = u((d_1, \cdots, d_k); x)  = \frac{Q}{W},
$$
and the utility for $(d_0, \cdots, d_l-1, d_{l+1}+1, \cdots, d_k)$ is
\begin{align*}
U' & = u((d_0, \cdots, d_l-1, d_{l+1}+1, \cdots, d_k); x) \\
    & = \frac{Q + G(i+1) (x- i x_{\min}) - G(i) (x- (i-1) x_{\min}) ) }{W + G(i+1) - G(i)} \\
    & = \frac{Q + xd^i  - x_{\min}(iG(i+1) -(i-1)G(i)) }{W + d^i}. \\
\end{align*}

To show $U' > U$, it equivalent to show
\[
\frac{xd^i  - x_{\min}(iG(i+1) -(i-1)G(i)) }{d^i} > \frac{Q}{W},
\]
or
\[
W\left(xd^i  - x_{\min}\left(iG(i+1) -(i-1)G(i)\right)\right) > d^i Q.
\]
Note that
\begin{eqnarray*}
  & &W\left(xd^i  - x_{\min}(iG(i+1) -(i-1)G(i))\right) \\
  & = &    W \left( xd^i   - x_{\min}( G(i+1) - (i-1) d^i) \right) \\
  & > &    W \left( kx_{\min}d^i   - x_{\min}( \frac{d^{i+1}}{d-1} - (i-1) d^i) \right) \\
  & = &  W x_{\min} d^i \left( k +\epsilon  - \frac{d}{d-1} - (i-1) \right) \\
  &\ge&   (\frac{1}{2} + \epsilon)W x_{\min} d^i,
\end{eqnarray*}
where the first inequality is because $x = (k+\epsilon)x_{\min}$ and
$$G(i+1) = \sum_{j = 0}^{i+1}d^j = \frac{d^{i+1}-1}{d-1} < \frac{d^{i+1}}{d-1};$$
the second inequality is because $k \ge i+1$ and $d \ge 3$.
Thus to show Theorem \ref{lem:tree:induction}, it suffices to show $(1+2\epsilon)Wx_{\min} > 2Q$.


\begin{claim}\label{claim:lem2:casek}
$(1+2\epsilon)Wx_{\min} > 2Q $.
\end{claim}
To prove Claim \ref{claim:lem2:casek}, we note that
\begin{eqnarray*}
&&(1+2\epsilon)W - \frac{2Q}{x_{\min}} \\
& = &  (1+2\epsilon)\left(\pi_0 + 1 + \sum_{j = 1}^{k} d_j G(j)\right) - 2\left(k+\epsilon + \sum_{j = 1}^{k} d_j G(j) (k + \epsilon - j + 1)\right)  \\
& \ge  &  \pi_0 - 2k + 1 - \sum_{i = 1}^{k} d_i G(i) (2k - 2i + 1)  \\ \
& > &  \pi_0 + 1 - 2k - dG(k)  \ge  0.
\end{eqnarray*}
The first equation is because $x = (k+\epsilon) x_{\min}$;
the first inequality is because $\epsilon \ge 0$;
the second inequality is because the following Claim \ref{eq:lem2:monotone};
and the last inequality is because $\pi_0 \ge dG(k) + 2k - 1$.

\begin{claim}\label{eq:lem2:monotone}
$\sum_{i = 1}^{k} d_i G(i) (2k - 2i + 1)  < dG(k)$.
\end{claim}
We prove Claim \ref{eq:lem2:monotone} in the appendix. \qed
\end{proof}

Lemma \ref{lem:tree:induction} is essentially a generalization of Lemma \ref{lem:no_bot} to tree structures,
which means if the information is already known to a sufficiently large number of players,
it is always beneficial for a player to make the information reach players in one deeper level.
By induction, we have the following corollary.

\begin{corollary} \label{cor:tree:fp}
Consider the sub-game on a single subtree $T$, where $r$ is the root player.
Assume $r$ receives the information with reward $x_r$ and $k \cdot x_{\min} \le x_r < (k+1)\cdot x_{\min}$ for some $k\ge 1$.
If $d \ge 3$, $\pi_0 \ge dG(k) + 2k - 1$, and all $r$'s descendants adopt  FP strategies,
$r$'s (unique) best strategy is $((k-1)\cdot x_{\min}, \cdots, (k-1)\cdot x_{\min})$.
\end{corollary}

Note that Lemmas \ref{lem:tree:descendant}, \ref{lem:tree:induction} and Corollary \ref{cor:tree:fp} do not only hold for the root players,
but also for any player in a subtree $T$. 

\subsection{Main Results} 

Given Lemma \ref{lem:tree:descendant} and Corollary \ref{cor:tree:fp}, it is not hard to verify that
FP strategy profile $\fs^*$ is a Nash equilibrium,
as for each fixed tree, the information is known to the other $f-1$ trees, and thus the number of aware players is at least
\[
(f-1)G(H+1) \ge (d-1)G(H+1) \ge (d+1)G(H).
\]
However, $\fs^*$ is not the unique Nash equilibrium.
Consider the following strategy profile.
Denote by 0 the root player of an arbitrary subtree $T$ and by $1,\cdots, d$ her children.
Consider the strategy profile $\fs'$: 
for all $x\geq x_{\min}$, $s_0(x) = (0,\cdots, 0)$ and $s_i(x) = (\bot, \cdots, \bot)$ for all $i \in \{1,\cdots, d\}$; 
all the other players in $T$ and all players not in $T$ play arbitrary strategies.
That is by notifying the players in depth 1, the root 0 withholds the complete reward no matter how much the initial reward is,
and the players in depth 1 do not propagate the information at all no matter how much reward the root player leaves for them.
Next, we claim strategy profile $\fs'$ is a Nash equilibrium and for simplicity, we assume $d \ge 4$.
First, for each player not in tree $T$, as there are at least $f-2$ other trees play FP, $\pi_0 \ge (f-2)G(H+1) > dG(k) + 2k -1$ for any $k \ge 0$.
By Corollary \ref{cor:tree:fp}, the best response of them is FP, thus all these players will not deviate.
Second, the root player 0 does not deviate, as all her children do not propagate the information and her best strategy is to withhold all the reward;
Finally, no player in $\{1,\cdots, d\}$ deviates, as the remaining reward for the information is 0.
We can observe that in $\fs'$, all $\{0, 1,\cdots, d\}$ played ``bad'' strategies.
By deviating to FP strategies simultaneously, all of them can improve their utilities,
which means $\fs^*$ is more stable equilibrium than $\fs'$.

We next define {\em connected coalition-proof Nash equilibria on graphs}, which are stronger than an arbitrary Nash equilibrium.
Let $\Gamma = (N, \fS, u)$ be any game with $n$ players, and for any $C\subseteq N$, denote $\fS_C = \times_{i\in C} S_i$.
Let $G = (N, E)$ be a graph defined on players with $(i,j) \in E$ meaning players $i$ and $j$ can communicate with each other directly.
A Nash equilibrium $\fs\in \fS$ is called {\em connected coalition-proof on $G$}
if there is no coalition $C \subseteq N$ with the induced subgraph of $C$ on $G$ being connected
such that by deviating to $\fs'_C$, $u_i(\fs'_C, \fs_{N\setminus C}) \ge u_i(\fs)$ for any $i \in C$
and there is one $j \in C$ such that $u_j(\fs'_C, \fs_{N\setminus C}) > u_j(\fs)$.
It is easy to see that any strong Nash equilibrium is connected coalition-proof on a complete graph
and any Nash equilibrium is connected coalition-proof on an empty graph.
As pointed out by Myerson in \cite{DBLP:journals/mor/Myerson77}, the requirement of connected coalition-proof is practical as
the players in any deviating coalition should be able to communicate. 

\begin{theorem}\label{thm:coordinate}
For $f\ge d\ge 3$, full propagation strategy profile $\fs^*$ is a connected coalition-proof Nash equilibrium on $\cG$.
\end{theorem}

\begin{proof}
For any coalition $C \subseteq N$, denote by $\cG(C)$ the induced subgraph of $C$ in $\cG$.
We first observe that $\cG(C)$ being connected implies $\cG(C)$ being a subtree in some $T$.
Then to prove the theorem, it suffices to show for any deviation $\fs'_C$ from $\fs^*_C$,
there is at least one player whose utility is smaller than the case when all of them play FP.
Let $r$ be the root of $\cG(C)$.
Without loss of generality, we reorder the players in $C\setminus \{r\}$ by $\{1,2, \cdots, c\}$ where $c = |C\setminus \{r\}|$.

Given the other players not in $T$ play FP strategies, the information will be known to at least $\pi_0$ players, and when $d \ge 3$,
by Equation \ref{eq:G(k)>=2k - 1},
$$\pi_0 \ge (f-1)G(H+1) > dG(H) + 2H -1.$$

If $r$ is in depth $H-1$, all players in $C\setminus \{r\}$ are dummy, whose actions do not affect $r$'s utility.
Thus by Lemma \ref{lem:no_bot}, any deviation from $\fs^*_C$ will strictly decreases $r$'s utility.
In the following, we assume $r$ is above depth $ H-1$.

Since $\pi_0 \ge dG(H) + 2H -1$, by Lemma \ref{lem:tree:descendant},
$$
u_r(\fs'_{C}, \fs^*_{N \setminus C}) \le u_r(\fs'_{C\setminus \{1\}}, \fs^*_{\{1\} \cup (N \setminus C)}).
$$
That is by changing player $1$'s strategy to FP, $r$'s utility can only increase.
We continue this procedure for player $i = 2, \cdots, c$ by changing her strategy from $s_i$ to $s^*_i$,
and we have
\begin{eqnarray*}
u_r(\fs'_{\{r, i, \cdots, c\}}, \fs^*_{\{1,\cdots, i-1\} \cup N \setminus C}) \le u_r(\fs'_{\{r, i +1, \cdots, c\}}, \fs^*_{\{1,\cdots, i\} \cup N \setminus C}).
\end{eqnarray*}
Eventually, we obtain
$$
u_r(\fs'_{C}, \fs^*_{N \setminus C}) \le u_r(\fs'_{r}, \fs^*_{N\setminus \{r\}}).
$$
By Corollary \ref{cor:tree:fp}, when all $i$'s descendants play FP strategies,
$$
u_r(\fs'_{r}, \fs^*_{N\setminus \{r\}}) < u_r(\fs^*).
$$
Thus, $u_r(\fs'_{C}, \fs^*_{N \setminus C}) <  u_r(\fs^*)$, which finishes the proof.
\qed
\end{proof}

By Theorem \ref{thm:coordinate},
we have shown that $\fs^*$ is more stable than an ordinary Nash equilibrium.
But $\fs^*$ is not a strong Nash equilibrium, because all the $f$ root players can form a deviating coalition such that none of them
propagates the information, which brings each root player utility $\frac{1}{f}$.
However, as we have seen in the introduction, no propagation of these $f$ players is not an equilibrium.

Next we show that $\fs^*$ survives in any possible order of elimination of dominated strategies,
and uniquely survives in an almost monotonic order of elimination of dominated strategies,
which surpasses the result in \cite{babaioff2012bitcoin}.
We call an elimination order  {\em monotonic} if for any player $i$ and any two eliminated strategies $s_i$ and $s'_i$,
$\min_j s_i(x^*_i)_{j} < \min_j s'_i(x^*_i)_{j}$ implies that $s_i$ is not eliminated after $s'_i$,
where $x^*_i$ is the minimum $x_i$ such that $\min_j s_i(x_i)_{j} \neq \min_j s'_i(x_i)_{j}$.
Assume $\bot < 0$.
We call an elimination order  {\em almost monotonic} if the condition is relaxed to
$s_i(x^*_i)_{j^*} < s'_i(x^*_i)_{j^*} + x_{\min}$ implying $s_i$  is not eliminated after $s'_j$.


\begin{theorem} \label{thm:coordinate:elimination}
For $f > d\ge 3$, $\fs^*$ survives in any possible order of elimination of dominated strategies.
Moreover, $\fs^*$ is the unique Nash equilibrium survives in an almost monotonic order of elimination of dominated strategies.
\end{theorem}

Theorem \ref{thm:coordinate:elimination} is proved in the appendix.
Note that all our results in this section hold as long as each player has at least $d$ children, where $d \ge 3$.
Actually, similar results also hold for the case of $d = 1, 2$, but the players may withhold multiples of $x_{\min}$.

\paragraph{Remark.}
Theorem \ref{thm:coordinate:elimination} is similar to the result in \cite{babaioff2012bitcoin}.
They designed a hybrid rewarding scheme, which combines two nearly-uniform algorithms, that is Sybil-proof.
Each player in a nearly-uniform algorithm
A nearly-uniform algorithm specifies a maximal length $H$ of rewarding path and rewards winning player in a chain of length $h$ a reward of $1 + \beta \cdot (H - h + 1)$.
All the players between the sender and the winner are rewarded $\beta$.
Though  $\beta$ plays a similar role with our $x_{\min}$, there are several differences.
One of the main differences is that, by the design of the free market, the players can arbitrarily decide her own charge as long as the charge is at least of $x_{\min}$,
which is not restricted to be integer multiples of $x_{\min}$ (i.e., Sybil copies).
A second difference is that we do not combine two different schemes.
The advantage of the free market design is that  the system does not need to (carefully) specify who gets how much,
and the players themselves are already motivated to maximally propagate the information,
which is also the main take-home message of the current paper.

\section{A Class of Non-tree Networks}
\label{sec:nontree}

In this section, we investigate the extent to which our results for trees can be extended to general networks.
In a non-tree network, a player may get the information from multiple neighbours and she will claim the one who leaves the highest reward to her,
where tie is broken arbitrarily but consistently.
Again let $\cG = (V, E)$ be an arbitrary network and $s \in V$ be the initial sender.
We first introduce the notions of  {\em good friends} and {\em best friends}.
For two players $i$ and $j$, $i$ is $j$'s best friend if (1) $i$ and $j$ are connected, and
(2) every shortest path from sender $s$ to $j$ passes $i$.
If $i$ is $j$'s best friend, then $j$ is called $i$'s good friend.
Note that each player can have at most one best friend but multiple good friends.
For example, in Fig \ref{fig:non-tree:friends}(a), $a,b$, and $d$ are the best friends of $c,e$, and $g$, respectively.
However, $d$ and $f$ do not have any best friend as each of them has two disjoint shortest paths to $s$.
Let $\cT=(V,E')$ be the subgraph of $\cG$, where $E'\subseteq E$ and $(i,j)\in E'$ if $i$ is $j$'s good or best friend.
Note that $\cT$ is a spanning forest of $\cG$, and the root of each tree is either $s$ or a player who does not have best friend.
As an example, the solid lines in Fig \ref{fig:non-tree:friends}(a) form a {\em good-friendship} graph.

\begin{figure}
    \centering

\subfigure[Good-Friendship Graph $\cT$.]{
\begin{minipage}[t]{0.5\linewidth}
\centering
\includegraphics[width=50mm]{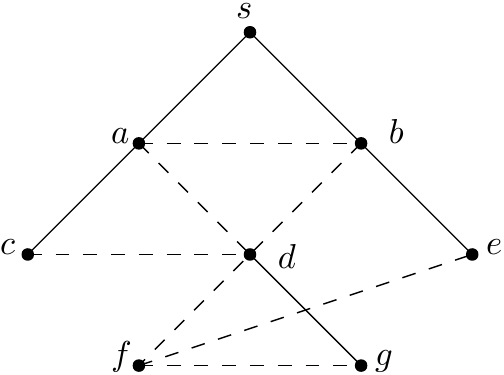}
\end{minipage}%
}%
\subfigure[Information Propagation Tree $T^*$.]{
\begin{minipage}[t]{0.5\linewidth}
\centering
\includegraphics[width=50mm]{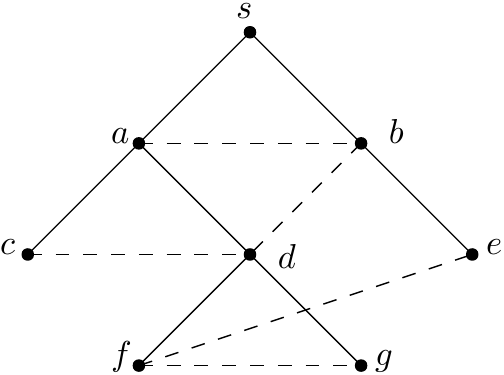}
\end{minipage}%
}%

    \caption{Illustration of Good-Friendship Graph and Information Propagation Tree.}
    \label{fig:non-tree:friends}
\end{figure}


Next we prove that if every player has at least three good friends, then full propagation is a Nash equilibrium.
Moreover, this equilibrium is robust to collective deviations of good friends.
Note that the case of $d$-ary tree with $d\ge 3$ in Section \ref{sec:tree} is a special case of this situation where $\cG = \cT$.

\begin{theorem}
\label{thm:nontree}
If every player has at least 3 good friends, then FP strategies $\fs^*$ is a Nash equilibrium.
Moreover, it is connected coalition-proof on $\cT$.
\end{theorem}

\begin{proof}
If every player has at least 3 good friends, then every node in $\cT$ has at least 3 children in its corresponding tree.
Let $T^*$ be the corresponding information propagation tree under $\fs^*$, where $i$ is connected to $j$ if $j$ claims $i$ as the ancestor.
Essentially, $T^*$ connects all the trees in $\cT$ as shown in Fig \ref{fig:non-tree:friends}(b).
We first prove $\fs^*$ is a Nash equilibrium. Suppose, for the sake of contradiction, some player $i$ wants to deviate from $\fs^*$.
We partition $i$'s neighbours into two sets $S_1$ and $S_2$,
where $S_1$ contains the players whose distance to $s$ is at most $i$'s distance to $s$,
and $S_2$ contains the players whose distance to $s$ is longer than $i$'s distance to $s$.
Note that player $i$'s action does not affect players in $S_1$ and their descendants in $T^*$.
We further partition $S_2$ into $S_{21}$ and $S_{22}$, where $S_{21}$ contains all $i$'s good friends and $S_{22} = S_2 \setminus S_{21}$.
If $i$ does not fully propagate to players in $S_{22}$, $i$ loses all the referring utility from $S_{22}$ because they can still be informed via other paths, which decreases $i$'s utility.
By the condition of the lemma, $|S_{21}| \ge 3$, and only players in $S_{21}$ are connected with $i$ in $\cT$.
Moreover, under $\fs^*$, $i$ will be claimed as the ancestor with higher priority by all her descendants in $\cT$.
Then our problem degenerates to the case of trees in the previous section, and  by Lemma \ref{lem:tree:induction},  $i$'s utility is maximized by FP, which means $\fs^*$ is a Nash equilibrium.
Moreover, the reason why $\fs^*$ is connected coalition-proof on $\cT$ is in the same with Theorem \ref{thm:coordinate}:
Any connected subgraph in $\cT$ forms a subtree with a root player $i$, whose utility decreases by deviating from $\fs^*$. \qed
\end{proof}

Here we note that $\fs^*$ may not be connected coalition-proof on the original network $\cG$.
To see this, if all senders' direct neighbours are connected with each other,
they can form a connected coalition and do not propagate the information to others, which brings each of them higher utility than full propagation.

\section{Experiments}
\label{sec:experiments}
Finally, we conduct experiments to confirm the validity of the free market design in random networks. 
We first introduce the parameters in our experiments.
Given a parameter $d$, there are $n$ players and each player is randomly connected to $\frac{d}{2}$ other players,
so that the expected degree of each player is $d$.
We study the utility of a fixed player and the sender is randomly selected from the other players.
Fix all other players' strategies to be full propagation.
Set the initial reward to $1$ and $x_{\min} = \frac{1}{H}$. 
In each experiment, we randomly generate $K$ networks. 
For simplicity, we only consider the strategies that withhold integral multiples of $x_{\min}$,
and calculate the expected utility of the studied player for each strategy.

\begin{figure}[ht!]
\centering
\includegraphics[width=90mm]{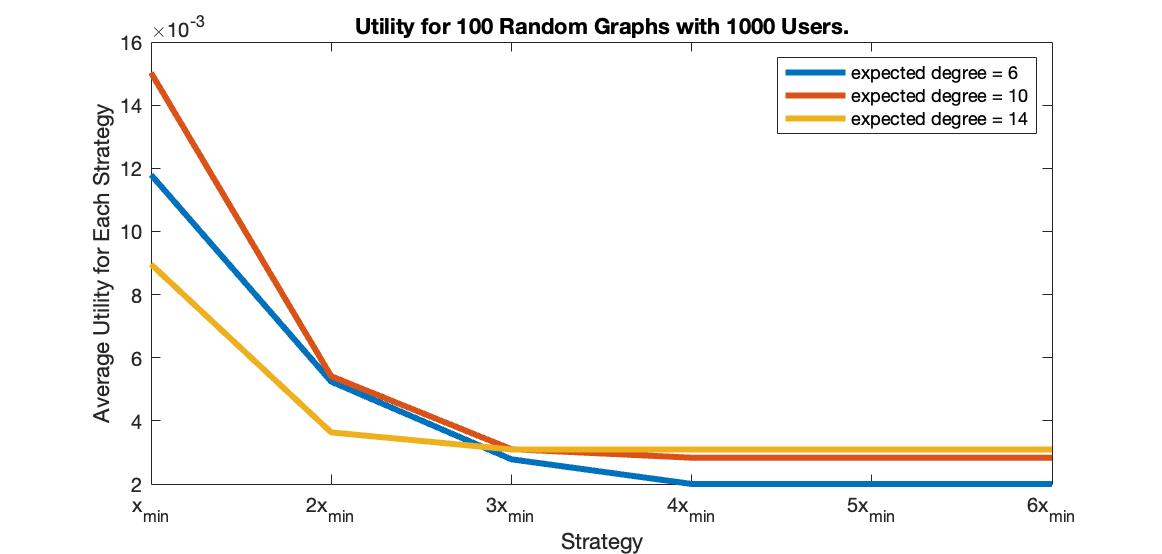}
\vspace{-3mm}
\caption{When all players have the same expected degree.}
\label{fig:exp1}
\end{figure}

In Fig \ref{fig:exp1}, we set $n = 1000$ and $H=6$.
For each $d = 6/10/14$, we randomly construct $K=100$ networks.
We observe that full propagation (by withholding $x_{\min}$)
brings a significantly higher utility on average than the other strategies (by withholding $kx_{\min}$ for $k >1$).
Then we revise the experiments by making the studied player more powerful or less powerful than the others,
where a powerful player has higher degree than the others and a powerless player has lower degree than the others.
In Fig \ref{fig:exp2}, the degree for the powerful player is set to $2d$ and for the powerless player it is set to be $\frac{d}{2}$.
Similarly, we observe that in both cases, full propagation brings significantly higher utility than the other strategies.

\begin{figure}[ht!]
\centering
\includegraphics[width=100mm]{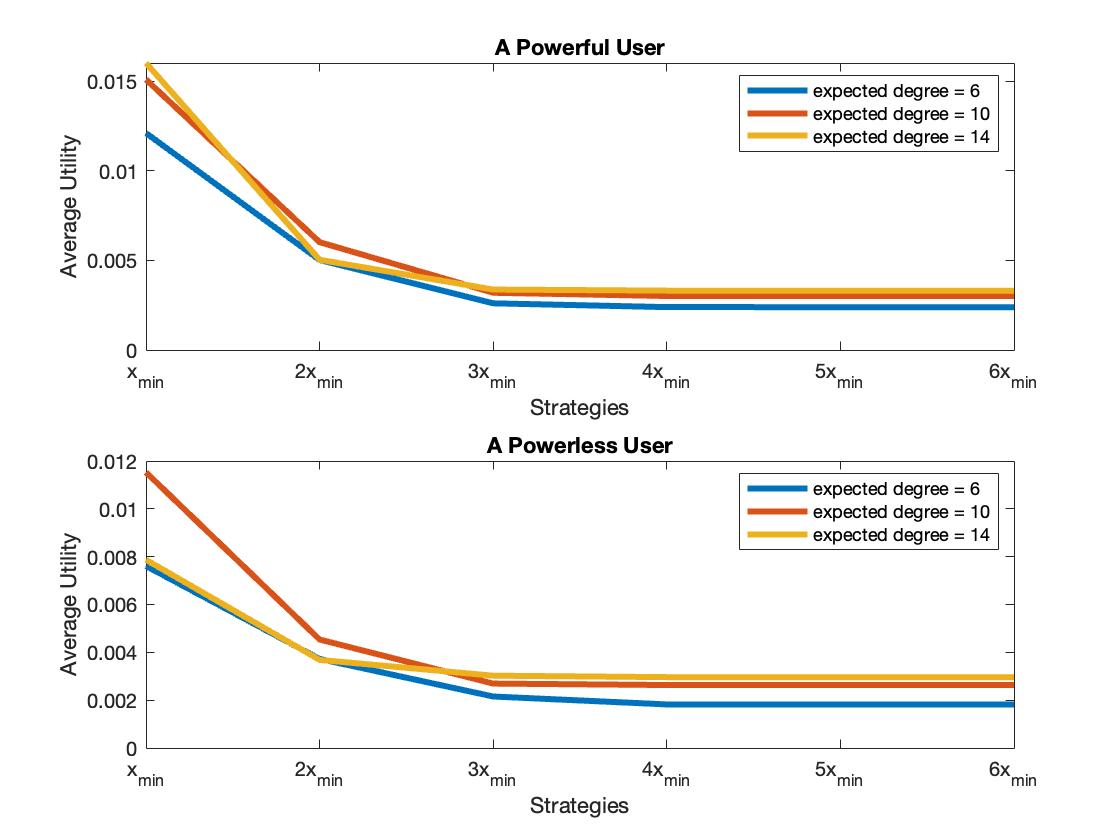}
\vspace{-3mm}
\caption{When the player is powerful or powerless.}
\label{fig:exp2}
\end{figure}


Note that full propagation is not always optimal.
If a player $i$ exclusively controls some players (i.e., the sender's information can only be reached to them via $i$)
and these players form a complete graph, the optimal local strategy for $i$ is to leave 0 to all of them instead of full propagation.
Thus, if the network structure is known to players, the players can compute their optimal propagating strategies, which may not be full propagation.
However,  such ``bad'' networks barely happen in random networks and never happen in reality.
To further avoid security threats, such as in a blockchain protocol, the players can be randomly re-connected periodically so that it is not beneficial for players to spend effort on learning network structures any more.

\section{Conclusion and Future Directions}
In this work, we design a free market with lotteries to incentivize players to maximally propagate information to their friends.
For trees and a large class of non-tree networks, we prove that full propagation is robust against connected coalitions of players.
Particularly, full propagation in tree networks uniquely survives in an interval-based monotone iterative elimination of dominated strategies.
For future directions, from a theoretical perspective, it is intriguing to analyze players' behaviours in the free market within arbitrary or random networks
and consider non-uniform lotteries. 
For example, if the property of every node having 3 good friends holds for certain random graphs, then Theorem \ref{thm:nontree}
holds for such graphs as well.
In addition, in the current market only 
the final winner and the players on the propagation path reaching her get rewarded.
It would be interesting to see if there is a way to select and reward more than one winners and corresponding propagation paths.
In a blockchain, this could correspond to multiple block proposers whose blocks are not eventually finalized,
and may help reducing the players' risk, so as to encourage more to participate in a risk-avert model.
Moreover, it is interesting and challenging to design new reward schemes to incentivize information propagation in more complex networks.
Finally, from a practical perspective, experiments on synthetic and real-world data for complex networks may
reveal important insights on the behavior of a free market in such networks, further enabling market design for them.

\appendix

\section*{Appendix}

\section{Missing Proofs}
\subsection{Proof of Lemma \ref{lem:no_bot}}

\noindent\textbf{Lemma \ref{lem:no_bot}} (restated){\bf .} {\em
For any player $i$,  let $\fs_{-i}$ be $-i$'s strategy profile such that the remaining reward $x_i$ to $i$ is at least  $x_{\min}$.
Let $s_i(x_i) = (z_1, \cdots, z_j, \cdots, z_{|NB_i|})$ be a strategy with $z_j = \bot$ for some $j$
and $s'_i(x_i) = (z_1, \cdots, z'_j, \cdots, z_{|NB_i|})$ be a new strategy by changing $j$'s action from $\bot$ to $0$.
Then $u_i(s_i, \fs_{-i}) < u_i(s'_i, \fs_{-i})$.
}


\begin{proof}
We first consider player $i$'s utility under strategy profile $(s_i, \fs_{-i})$.
For $l\in [|NB_i|]$, let $\Delta_l \ge 0$ be the number of players $i$ successfully refers through neighbour $l$
and $\pi_0$ be the total number of players who are aware of the information.
Then $i$'s utility is
$$ u_i(s_i, \fs_{-i}) = \frac{x_i + \sum_{z_l \neq \bot} (x_i - z_l)\Delta_l}{\pi_0}.$$

Now consider $i$'s utility under strategy profile $(s'_i, \fs_{-i})$.
By changing $z_j = \bot$ to $z_j = 0$, $j$ may or may not claim $i$ as her ancestor.
If not, it means that $j$ is already aware of the information, but someone else leaves a higher utility for her, thus
$i$'s utility does not change, i.e., $u_i(s'_i, \fs_{-i}) = u_i(s_i, \fs_{-i})$.

If $j$ claims $i$ as her ancestor, there are two cases:
(1) $j$ is already aware of the information, but $i$ is on the shortest path or the tie breaking rule is in favour of $j$, thus $s'_i$ does not increase $\pi_0$;
(2) $j$ is not aware of the information, thus $s'_i$ increase $\pi_0$ by 1.
In any case, $i$'s utility is at least
$$u_i(s'_i, \fs_{-i}) = \frac{2x_i + \sum_{z_l \neq \bot} (x_i - z_l)\Delta_l}{\pi_0 + 1}.$$
To show $u_i(s'_i, \fs_{-i}) \ge u_i(s_i, \fs_{-i})$, it suffices to see the following inequalities.
\begin{eqnarray*}
&& \frac{2x_i + \sum_{z_l \neq \bot} (x_i - z_l)\Delta_l}{\pi_0 + 1} -
\frac{x_i + \sum_{z_l \neq \bot} (x_i - z_l)\Delta_l}{\pi_0} \\
& = & \frac{1}{\pi_0(\pi_0 + 1)} \bigg(2\pi_0 x_i + \pi_0 \sum_{z_l \neq \bot} (x_i - z_l)\Delta_l  \\
&& - (\pi_0 + 1)x_i  - (\pi_0 +1)\sum_{z_l \neq \bot} (x_i - z_l)\Delta_l \bigg) \\
& = & \frac{(\pi_0-1) x_i - \sum_{z_l \neq \bot} (x_i - z_l)\Delta_l}{\pi_0(\pi_0 + 1)} \\
& > & \frac{(\pi_0-1) x_i - (\pi_0 - 1) x_i}{\pi_0(\pi_0 + 1)} = 0.
\end{eqnarray*}
The inequality above is because $\sum_l \Delta_l < \pi_0 - 1$ and $x_i - z_l\le x_i$.
Thus $u_i(s'_i, \fs_{-i}) \ge u_i(s_i, \fs_{-i})$ and we complete the proof.\qed
\end{proof}

\subsection{Proof of Claim \ref{eq:lem2:monotone}}

\noindent\textbf{Claim \ref{eq:lem2:monotone}} (restated){\bf .}  {\em
$\sum_{i = 1}^{k} d_i G(i) (2k - 2i + 1)  < dG(k)$.
}

\begin{proof}
Let
\begin{eqnarray*}
f(y) & = &  G(\lfloor y \rfloor) (2k - 2y + 1) = \frac{d^{\lfloor y \rfloor} - 1}{d-1} (2k - 2y + 1)\\
& \le & \frac{d^{y} - 1}{d-1} (2k - 2y + 1) = \bar{f}(y).
\end{eqnarray*}
Note that as $\sum_{i = 1}^{k} d_i = d$, to prove the claim,
we only need to show $\bar{f}(y) < G(k) = \bar{f}(k)$.
It is not hard to see when $y \le k-1$ and $d \ge 3$,
$$
\bar{f}'(y) =  \frac{d^y}{d-1} \left( (2k-2y+1) \ln{d} -2   \right) + \frac{2}{d-1} > 0.
$$
That is $\bar{f}(y) < \bar{f}(k-1)$ for all $y < k-1$.
Then we are left to compare $\bar{f}(k)$ and $\bar{f}(k-1)$,
$$
\bar{f}(k) =  \frac{d^k - 1}{d-1} >  \frac{d\cdot d^{k-1} - 3}{d-1} \ge  3\frac{d^{k-1} - 1}{d-1} = \bar{f}(k-1),
$$
where the first inequality is also because of $d \ge 3$.
This finishes the proof of Claim \ref{eq:lem2:monotone}.
\qed
\end{proof}

\subsection{Proof of Theorem \ref{thm:coordinate:elimination}}
\label{sec:tree:descendant}

\noindent\textbf{Theorem \ref{thm:coordinate:elimination}} (restated){\bf .} {\em
For $f > d\ge 3$, $\fs^*$ survives in any possible order of elimination of dominated strategies.
Moreover, $\fs^*$ is the unique Nash equilibrium survives in an almost monotonic order of elimination of dominated strategies.
}

\begin{proof}
We prove the first part via contradiction. 
For any order of elimination of dominated strategies, consider the first time that some player $i$'s FP strategy $s_i^*$ is eliminated.
Note that $s_i^*$ being dominated by $s_i'$ means for all $\fs_{-i}$, $s_i'$ brings utility as least as much as $s_i^*$.
Particularly, $u_i(s_i^*, \fs_{-i}^*) \le u_i(s_i', \fs_{-i}^*)$.
However, given all other players play FP strategies, the information will be known to least $(f-1)G(H+1)$ players.
Thus, by Corollary \ref{cor:tree:fp}, $u_i(s_i^*, \fs_{-i}^*) > u_i(s_i', \fs_{-i}^*)$, where we reach the contradiction.
Next, we prove the second part.

For every player $i$, when the reward $x_i$ is clear, we denote by $s_i(x_i) = (z_1,\cdots,z_d)$. 
If $x_i < x_{\min}$, it is not possible for $i$ to propagate the information, and
the unique action is $s_i(x_i) = (\bot,\cdots, \bot)$, thus we assume $x_i \ge x_{\min}$.
In the following, we explicitly give the elimination order of dominated strategies.

Let $\pi_0$ be the number of players who are aware of the information.
As all root players know the information, by default, $\pi_0 \ge d + 1$.

\paragraph{Round 1.}
For any $x_i \ge x_{\min}$,
\begin{align*}
s_i(x_i) = (z_1, \cdots, 0< z_j < x_{\min}, \cdots, z_d) \preceq  s'_i(x_i) = (z_1, \cdots, 0, \cdots, z_d).
\end{align*}
Note that leaving less than $x_{\min}$ to $j$, $j$ is not able to propagate the information.
That is for any $\fs_{-i}$, under both strategy profiles $(s_i, \fs_{-i})$ and $(s'_i, \fs_{-i})$, the numbers of aware players are the same.
As $s'_i$ potentially brings higher referring utility, $u_i(s_i, \fs_{-i}) \le u_i(s'_i, \fs_{-i})$.
Particularly, if $\fs_{-i}$ is the one that $i$ receives exactly $x_i$ from her parent\footnote{If $i$ is in depth 0, we consider $x_i$ as the initial reward from the sender.},
$u_i(s_i, \fs_{-i}) < u_i(s'_i, \fs_{-i})$.
Thus $s_i(x_i) \preceq s'_i(x_i)$.
Eliminate all such dominated strategies for all players.

If $x_{\min} \le x_i < 2x_{\min}$, as $0 \le z_j \le x_i - x_{\min} < x_{\min}$ for all $j$,
for any remaining strategy $s_i$, $s_i(x_i) = (0,\cdots, 0)$.
That is when the remaining reward to a player is less than $2x_{\min}$, she will fully propagate the information.

As all root players are aware of the information and all of them will notify the players in depth 1, after this round,
for all possible strategy profiles, $\pi_0 \ge fG(2) \ge (d+1)G(2)$.




\paragraph{Round 2.}
First,  for any $x_i \ge 2x_{\min}$,
\begin{align*}
s_i(x_i) = (z_1, \cdots, z_j = 0, \cdots, z_d) \preceq  s'_i(x_i) =  (z_1, \cdots, x_{\min}, \cdots, z_d).
\end{align*}
This is because after the first round, all players who receives $x_{\min}$ will propagate the information to the next level,
i.e. full propagation. By Corollary \ref{cor:tree:fp} and the fact that $\pi_0 \ge (d+1)G(2) > G(2) + 3$, $s_i(x_i) \preceq s'_i(x_i)$.
Eliminate all such dominated strategies for all players. 

By induction, we have the following.

\paragraph{Round $l > 2$.}
First, for any $x_i \ge lx_{\min}$ and $g\in \{0,1,\cdots, l-2\}$,
\begin{align*}
s_i(x_i) = (z_1, \cdots, z_j = gx_{\min}, \cdots, z_d)  \preceq  s'_i(x_i) =  (z_1, \cdots, (l-1)x_{\min}, \cdots, z_d).
\end{align*}
To see why, note that after the first $l-1$ rounds, $\pi_0 \ge fG(l) \ge (d+1)G(l) > dG(l) + 2l -1$ and if $j$ receives $(l-1)x_{\min}$,
she and her descendants fully propagate the information.
By Corollary \ref{cor:tree:fp}, $s_i(x_i) \preceq s'_i(x_i)$.
Eliminate all such dominated strategies for all players. 

Second, for any $x_i \ge l x_{\min}$,
\begin{align*}
 s_i(x_i) =(z_1, \cdots, z_j, \cdots, z_d)  \preceq   s'_i(x_i) =(z_1, \cdots, (l-1)x_{\min}, \cdots, z_d),
\end{align*}
where $(l-1)x_{\min} < z_j < l x_{\min}$.
This is because after round $l-1$, if $(l-1)x_{\min} \le z_j < l x_{\min}$, by induction,
child $j$'s only possible action is $s_j(z_j)=((l-2) x_{\min},\cdots,(l-2) x_{\min})$.
Then leaving $ z_j $ to child $j$,  $j$  will not notify more nodes than leaving $(l-1) x_{\min}$ to him.
That is for any $\fs_{-i}$, under both strategy profiles $(s_i, \fs_{-i})$ and $(s'_i, \fs_{-i})$, the numbers of aware players are the same.
As $s'_i$ potentially brings higher referring utility, $u_i(s_i, \fs_{-i}) \le u_i(s'_i, \fs_{-i})$.
Moreover, if $\fs_{-i}$ is the one such that $i$ receives exactly $x_i$ from her parent, $u_i(s_i, \fs_{-i}) < u_i(s'_i, \fs_{-i})$.
Thus $s_i(x_i) \preceq s'_i(x_i)$.
Eliminate all such dominated strategies for all players.

Note that if $lx_{\min} \le x_i < (l+1)x_{\min}$, as $0 \le z_j \le x_i - x_{\min} < lx_{\min}$ for any $j$,
for any remaining strategy $s_i$ after the first $l$ rounds, $s_i(x_i) = ((l-1)x_{\min},\cdots, (l-1)x_{\min})$.
Combining with previous rounds, when the remaining reward to a player is less than $(l+1)x_{\min}$, she will fully propagate the information.

Moreover, after this round, all the players in depth $l+1$ are also aware of the information,
$\pi_0 \ge fG(l+1) \ge (d+1)G(l+1) \ge dG(l+1) + 2(l+1) -1$.

From above analysis, we observe that for $kx_{\min}\le x_i < (k+1)x_{\min}$,
the only remaining strategy is $s_i(x_i) = ((k-1)x_{\min}, \cdots, (k+-1)x_{\min})$, i.e., full propagation.
Moreover, the elimination order of dominated strategies is almost monotonic. \qed
\end{proof}

%
%
%
%
\bibliographystyle{splncs04}
\bibliography{ref}

\begin{thebibliography}{10}
\providecommand{\url}[1]{\texttt{#1}}
\providecommand{\urlprefix}{URL }
\providecommand{\doi}[1]{https://doi.org/#1}

\bibitem{abraham2016solidus}
Abraham, I., Malkhi, D., Nayak, K., Ren, L., Spiegelman, A.: Solidus: An
  incentive-compatible cryptocurrency based on permissionless byzantine
  consensus  (2016)

\bibitem{arcaute2007threshold}
Arcaute, E., Kirsch, A., Kumar, R., Liben{-}Nowell, D., Vassilvitskii, S.: On
  threshold behavior in query incentive networks. In: {EC}. pp. 66--74. {ACM}
  (2007)

\bibitem{babaioff2012bitcoin}
Babaioff, M., Dobzinski, S., Oren, S., Zohar, A.: On bitcoin and red balloons.
  In: {EC}. pp. 56--73. {ACM} (2012)

\bibitem{bernheim1987coalition}
Bernheim, B.D., Peleg, B., Whinston, M.D.: Coalition-proof nash equilibria i.
  concepts. Journal of economic theory  \textbf{42}(1),  1--12 (1987)

\bibitem{cebrian2012finding}
Cebri{\'{a}}n, M., Coviello, L., Vattani, A., Voulgaris, P.: Finding red
  balloons with split contracts: robustness to individuals' selfishness. In:
  {STOC}. pp. 775--788. {ACM} (2012)

\bibitem{chen2010mobicent}
Chen, B.B., Chan, M.C.: Mobicent: a credit-based incentive system for
  disruption tolerant network. In: {INFOCOM}. pp. 875--883. {IEEE} (2010)

\bibitem{chenRelay}
Chen, J., Gilad, Y.: A relay market for the algorand blockchain. Manuscript
  (2018)

\bibitem{chen2013sybil}
Chen, W., Wang, Y., Yu, D., Zhang, L.: Sybil-proof mechanisms in query
  incentive networks. In: {EC}. pp. 197--214. {ACM} (2013)

\bibitem{chen2019axiomatic}
Chen, X., Papadimitriou, C.H., Roughgarden, T.: An axiomatic approach to block
  rewards. In: {AFT}. pp. 124--131. {ACM} (2019)

\bibitem{marmolejo2019competing}
Coss{\'{\i}}o, F.J.M., Brigham, E., Sela, B., Katz, J.: Competing
  (semi-)selfish miners in bitcoin. In: {AFT}. pp. 89--109. {ACM} (2019)

\bibitem{douceur2002sybil}
Douceur, J.R.: The sybil attack. In: {IPTPS}. Lecture Notes in Computer
  Science, vol.~2429, pp. 251--260. Springer (2002)

\bibitem{drucker2012simpler}
Drucker, F., Fleischer, L.: Simpler sybil-proof mechanisms for multi-level
  marketing. In: {EC}. pp. 441--458. {ACM} (2012)

\bibitem{emek2011mechanisms}
Emek, Y., Karidi, R., Tennenholtz, M., Zohar, A.: Mechanisms for multi-level
  marketing. In: {EC}. pp. 209--218. {ACM} (2011)

\bibitem{DBLP:conf/marble/ErsoyEL19}
Ersoy, O., Erkin, Z., Lagendijk, R.L.: Decentralized incentive-compatible and
  sybil-proof transaction advertisement. In: {MARBLE}. pp. 151--165. Springer
  (2019)

\bibitem{ersoy2018transaction}
Ersoy, O., Ren, Z., Erkin, Z., Lagendijk, R.L.: Transaction propagation on
  permissionless blockchains: Incentive and routing mechanisms. In: {CVCBT}.
  pp. 20--30. {IEEE} (2018)

\bibitem{Eyal2014selfishmining}
Eyal, I., Sirer, E.G.: Majority is not enough: Bitcoin mining is vulnerable.
  In: Christin, N., Safavi-Naini, R. (eds.) Financial Cryptography and Data
  Security. pp. 436--454 (2014)

\bibitem{johnson2014game}
Johnson, B., Laszka, A., Grossklags, J., Vasek, M., Moore, T.: Game-theoretic
  analysis of ddos attacks against bitcoin mining pools. In: International
  Conference on Financial Cryptography and Data Security. pp. 72--86 (2014)

\bibitem{kiayias2016blockchain}
Kiayias, A., Koutsoupias, E., Kyropoulou, M., Tselekounis, Y.: Blockchain
  mining games. In: {EC}. pp. 365--382. {ACM} (2016)

\bibitem{kleinberg2005query}
Kleinberg, J.M., Raghavan, P.: Query incentive networks. In: {FOCS}. pp.
  132--141. {IEEE} Computer Society (2005)

\bibitem{koutsoupias2019blockchain}
Koutsoupias, E., Lazos, P., Ogunlana, F., Serafino, P.: Blockchain mining games
  with pay forward. In: {WWW}. pp. 917--927. {ACM} (2019)

\bibitem{lewenberg2015bitcoin}
Lewenberg, Y., Bachrach, Y., Sompolinsky, Y., Zohar, A., Rosenschein, J.S.:
  Bitcoin mining pools: {A} cooperative game theoretic analysis. In: {AAMAS}.
  pp. 919--927. {ACM} (2015)

\bibitem{DBLP:journals/mor/Myerson77}
Myerson, R.B.: Graphs and cooperation in games. Math. Oper. Res.
  \textbf{2}(3),  225--229 (1977)

\bibitem{bitcoin2008}
Nakamoto, S.: Bitcoin: A peer-to-peer electronic cash system (2008)

\bibitem{Rosenfeld2011bwa}
Rosenfeld, M.: Analysis of bitcoin pooled mining reward systems. CoRR
  \textbf{abs/1112.4980} (2011)

\bibitem{schrijvers2016incentive}
Schrijvers, O., Bonneau, J., Boneh, D., Roughgarden, T.: Incentive
  compatibility of bitcoin mining pool reward functions. In: Financial
  Cryptography. Lecture Notes in Computer Science, vol.~9603, pp. 477--498.
  Springer (2016)

\bibitem{DBLP:conf/ecai/ShiZSWZ20}
Shi, H., Zhang, Y., Si, Z., Wang, L., Zhao, D.: Maximal information propagation
  with budgets. In: {ECAI}. Frontiers in Artificial Intelligence and
  Applications, vol.~325, pp. 211--218. {IOS} Press (2020)

\end{thebibliography}

\end{document}